\def\year{2021}\relax
\documentclass[letterpaper]{article} 
\usepackage{aaai21}  
\usepackage{times}  
\usepackage{helvet} 
\usepackage{courier}  
\usepackage[hyphens]{url}  
\usepackage{graphicx} 

\usepackage{times}
\usepackage{url}
\usepackage{graphicx}
\usepackage{amsmath,amsfonts,mathtools}
\usepackage{booktabs}
\usepackage{algorithm}
\usepackage{algorithmic}
\usepackage{paralist}
\usepackage{color}
\usepackage{complexity}
\usepackage{subfigure}

\newtheorem{theorem}{Theorem}[section]
\newtheorem{example}{Example}[section]
\newtheorem{definition}{Definition}[section]
\newtheorem{lemma}{Lemma}[section]
\newtheorem{corollary}{Corollary}[section]

\newcommand{\qed}{\unskip\hspace*{1em}\hspace{\fill}$\Box$}
\newenvironment{proof}[1][Proof]{\begin{trivlist}
		\item[\hskip \labelsep {\it #1:}]}{%
		\qed\end{trivlist}}
\newenvironment{proofof}[1]{{\vspace*{5pt} \noindent\bf Proof of #1:  }}{\hfill\rule{2mm}{2mm}\vspace*{5pt}}

\newcommand{\dma}{\textsf{DCA}\space}
\newcommand{\opt}{\mathsf{OPT}}
\newcommand{\alg}{\mathsf{ALG}}
\newcommand{\lb}{\theta}

\usepackage{caption} 
\title{Defending against Contagious Attacks on a Network \\
	with Resource Reallocation\footnote{Funded by the Science and Technology Development Fund, Macau SAR (File no. SKLIOTSC-2018-2020), the Start-up Research Grant of University of Macau (File no. SRG2020-00020-IOTSC). This work was supported in part by the Science and Technology Development Fund, Macau SAR under File no. 0060/2019/A1, and in part by Research Grant of University of Macau under Grant MYRG2018-00237-FST.}}

\author{
		Rufan Bai\textsuperscript{\rm 1},
		Haoxing Lin\textsuperscript{\rm 1},
		Xinyu Yang\textsuperscript{\rm 1},
		Xiaowei Wu\textsuperscript{\rm 1},
		Minming Li\footnote{City University of Hong Kong Shenzhen Research Institute, Shenzhen, P. R. China. The work described in this paper was partially sponsored by Project 11771365 supported by NSFC.}\textsuperscript{\rm 2},
		Weijia Jia
		\footnote{BNU-UIC Institute of Artificial Intelligence and Future Networks, Beijing Normal University (Zhuhai), Guangdong, China. The work was partially supported by Chinese National Research Fund (NSFC) Key Project No. 61532013; NSFC grant No. 61872239; and Guangdong Provincial Key Lab of AI and Multi-modal Data Processing at BNU-HKBU UIC.}\textsuperscript{\rm 3}\\
}

\affiliations{
		\textsuperscript{\rm 1}IoTSC, University of Macau,
		\textsuperscript{\rm 2}City University of Hong Kong,
		\textsuperscript{\rm 3}BNU (Zhuhai) \& UIC\\
		\{yb97439, starklin, mb95466, xiaoweiwu\}@um.edu.mo,
		minming.li@cityu.edu.hk,
		jiawj@sjtu.edu.cn
}

\begin{document}

	
	\maketitle
	
	\begin{abstract}
		In classic network security games, the defender distributes defending resources to the nodes of the network, and the attacker attacks a node, with the objective to maximize the damage caused.
		Existing models assume that the attack at node $u$ causes damage only at $u$.
		However, in many real-world security scenarios, the attack at a node $u$ spreads to the neighbors of $u$ and can cause damage at multiple nodes, e.g., for the outbreak of a virus.
		In this paper, we consider the network defending problem against contagious attacks.
		
		Existing works that study shared resources assume that the resource allocated to a node can be shared or duplicated between neighboring nodes.
		However, in real world, sharing resource naturally leads to a decrease in defending power of the source node, especially when defending against contagious attacks.
		To this end, we study the model in which resources allocated to a node can only be transferred to its neighboring nodes, which we refer to as a reallocation process.
		
		We show that this more general model is difficult in two aspects: (1) even for a fixed allocation of resources, we show that computing the optimal reallocation is \NP-hard; (2) for the case when reallocation is not allowed, we show that computing the optimal allocation (against contagious attack) is also \NP-hard.
		For positive results, we give a mixed integer linear program formulation for the problem and a bi-criteria approximation algorithm.
		Our experimental results demonstrate that the allocation and reallocation strategies our algorithm computes perform well in terms of minimizing the damage due to contagious attacks.
	\end{abstract}

	\section{Introduction}
	In recent years, security games have attracted much research attention within the artificial intelligence community and have been widely adopted for the computation of optimal allocation of security resources in many areas of the field~\cite{sagt/LetchfordCM09,daglib/0040483,aamas/YinT12,ijcai/SinhaFAKT18}. 
	A considerable portion of these works consider the security games played within a network structure, i.e., the network security games~\cite{assimakopoulos1987network,aaai/GanAVG17,ijcai/ZhangATWGJ17,atal/SchlenkerTXFTTV18}.
	In a network security game, there is an underlying graph, where each node of the graph represents a target with a defending requirement and a value to protect.
	The game is played between a defender who allocates defensive resources to the nodes of the graph and an attacker who picks a node to attack, depending on how the nodes are defended.
	
	Many existing works consider the setting when the allocated resource can be shared between neighboring nodes~\cite{ijcai/YinXGAJ15}.
	For example, Gan et al.~\cite{aaai/GanAV15} considered a network security game in which allocating one unit of resource to some target protects not only the target but also the neighboring targets.
	Li et al.~\cite{aaai/LiTW20} studied the model in which the defending power of each node $u$ is determined by the resource $r_u$ allocated to $u$, plus a linear function of the resources allocated to its neighbors.
	These models are mainly motivated by surveillance or patrolling applications, in which when a node $u$ shares resource with its neighbor, we do not need to worry about the defending power of $u$.
	
	However, for defending problems in which the attack is contagious, it is necessary to take into account the decrease in the defending power of node $u$, especially when $u$ is at the risk of being involved in the attack.
	Consider a contagious attack, e.g., the spread of a virus, on a node $v$.
	Suppose the attack spreads to the neighbors of $v$ and can cause damage at each of the nodes the attack spreads to, depending on how well the node is defended.
	In this case, if we measure the defending power of $v$ by taking into account the resources shared from its neighbor $u$, then naturally, we need to consider the decrease in the defending power of $u$.
	
	Ideally, a node $u$ can only \emph{transfer} (a fraction of) the resource it owns to its neighbor $v$, which increases the defending power of the receiver $v$ but decreases its own defending power.
	When defending against attacks without spreading effects, this assumption is equivalent to being able to duplicate resources between neighbor nodes, as we can always transfer the maximum possible resources towards the node under attack.
	However, when the attack can spread to neighbors of the node under attack, this assumption demands a stronger defending requirement.
	Specifically, the following example shows that when resources can only be transferred (instead of being duplicated), the total resource required to obtain a good defending result can be much larger.
	
	\begin{example}\label{example:transfer-vs-share}
		Consider a star graph, with node $u$ in the center, and $v_1,v_2,\ldots,v_{n-1}$ being neighbors of $u$.
		Suppose each node requires $1$ unit of resource to defend himself.
		Suppose node $u$ is attacked and the attack spreads to all neighbors of $u$.
		When resources can be duplicated, allocating one unit of resource at node $u$ guarantees that every node is sufficiently defended, and thus no loss is incurred.
		However, when resources can only be transferred, as long as the total resources allocated are less than $n$ units, there always exists at least one insufficiently defended node.
	\end{example}
	
	In the paper, we consider the problem of defending against contagious attack, in which the defending resources can only be transferred between neighboring nodes.
	Specifically, when the attacker attacks a node $u$ in the network, the attack spreads to neighbors of $u$ and may cause damage at multiple nodes.
	The defender decides an allocation strategy of defending resources to nodes in the graph before the attack happens, and is allowed to transfer some resources between neighboring nodes (subject to some capacity constraints) when the attack happens.
	Our model is motivated by real-world applications like defending against virus spreading.
	In these applications, it is reasonable to assume that we can transfer medical resources or doctors between neighboring cities or countries in order to minimize the damage when the virus breaks out.
	Unfortunately, existing models fail to capture such applications as most of them do not consider the reallocation of defending resources.

	\subsection{Our Results}
	
	We study the problem of computing optimal allocations and reallocations of defending resources.
	Since our main motivation of the problem is defending against virus spreading and, in real world, the allocation of defending resources is usually public information, we focus only on pure strategies, i.e., deterministic defending algorithms.
	We propose a mathematical model that generalizes that of~\cite{aaai/GanAV15,aaai/LiTW20}, and assume that (1) an attack spreads to a subset of nodes and may cause damage at each of them; (2) defending resources can be transferred between neighboring nodes, which we refer to as a reallocation of resources.
	%
	The objective is to minimize the maximum possible damage due to an attack.
	
	We show that this general model is difficult in two aspects.
	We first show that even with a given allocation of resources and a node that is attacked, computing the optimal reallocation is \NP-hard (Section~\ref{sec:reallocation}).
	Then we show that if no reallocation is allowed, the problem of computing the optimal allocation strategy is also \NP-hard (Section~\ref{ssec:hardness-allocation}).
	
	Regarding positive results, we provide mixed integer linear programs (MILPs) to model the computation of allocation and reallocation strategies (in Section~\ref{ssec:MILP-main}).
	We show that the optimal solutions for the MILPs provide optimal allocation and reallocation strategies.
	Since solving an MILP is not guaranteed to terminate in polynomial time, we also propose polynomial time algorithms for special cases and approximation algorithms.
	We give a polynomial time algorithm that decides whether there exists a defending strategy in which no loss incurs, and outputs one if it exists (Section~\ref{ssec:perfect-allocation}).
	Then we give a polynomial time bi-criteria $(\frac{1}{1-\epsilon},\frac{1}{\epsilon})$-approximation algorithm, for any $\epsilon\in(0,1)$ (see Section~\ref{ssec:bicriteria-approx} for a formal definition of bi-criteria approximations).
	Specifically, for $\epsilon=0.5$ we have a bi-criteria $(2,2)$-approximation.
	Moreover, we show that under the Unique Game Conjecture~\cite{jcss/KhotR08}, there does not exist $(2-\delta,2-\delta)$-approximation, for any constant $\delta > 0$.
	
	Finally, we extensively evaluate our algorithms on synthetic and real-world datasets in Section~\ref{sec:experiments}.
	
	\subsection{Other Related Work}
	As mentioned, there is a sequence of existing works in the network security game domain that consider resource sharing between nodes.
	Gan et al.~\cite{aaai/GanAV15,aaai/GanAVG17} consider models in which allocating a unit of defending resource to a node can also protect the neighbors of that node. 
	Their models only study the binary version of resource allocation, i.e., $r_u\in\{0,1\}$.
	Yin et al.~\cite{ijcai/YinXGAJ15} also study a model in which the resource can be shared, and they assume sharing resources takes time.
	However, these existing models does not consider the contagious attacks or the resource reallocation.
	
	There are also works that study the contagion in network security games~\cite{gamenets/NguyenAB09,allerton/BachrachDG13,aamas/VorobeychikL15,jet/AcemogluMO16,expert/LouSV17,goyal2014attack,aspnes2006inoculation}.
	Besides, Tsai et al.~\cite{aaai/TsaiNT12} study a zero-sum two-player influence blocking maximization game, in which the attacker and the defender try to maximize their influence on a network.
	However, these works do not model the problem in terms of allocating defending resources to meet defending requirements and minimizing the loss due to attack, and thus are incomparable to our model.
	There are other works that study contagion of attack by assuming that an insufficiently protected node can affect the defending result of its neighboring nodes~\cite{games/ChanCO17,aaai/LiTW20}.
	There are also works that study game-theoretic models of the security games~\cite{kunreuther2003interdependent,gamesec/JohnsonGCC10,uai/ChanCO12}.

	\section{Model Description}
	
	We model the network as an undirected\footnote{While we assume the graph is undirected, it can be verified that all our results extend straightforwardly to directed graphs.} connected graph $G(V,E)$, where each node $u \in V$ has a \emph{threshold} $\lb_u$ that represents the defending requirement, and a \emph{value} $\alpha_u$ that represents the possible damage due to an attack at node $u$.
	We use $N(u):= \{v: (u,v)\in E\}$ to denote the set of neighbors for node $u\in V$.
	We use $N_k(u)$ to denote the set of nodes at distance at most $k$ from $u\in V$.
	By definition we have $N_1(u) = \{u\}\cup N(u)$.
	We use $n$ and $m$ to denote the number of nodes and edges in the graph $G$, respectively.

	\subsection{Defending Resource and Defending Power}
	
	The defender has a total resource of $R$ that can be distributed to nodes in $V$, where $r_u$ is the \emph{defending resource}\footnote{Similar to~\cite{aaai/LiTW20}, we assume the resource can be allocated continuously in our model.} allocated to node $u$, and $\sum_{u\in V}r_u = R$.
	Each node $u$ can transfer at most $w_{uv}\cdot r_u$ units of defending resource to each of its neighbor $v$, where $w_{uv}\in [0,1]$ is the \emph{weight} of edge $(u,v)$, which represents the efficiency (or willingness) when transferring defending resource between $u$ and $v$.
	
	\begin{definition}[Allocation Strategy]
		We use $r_u\geq 0$ to denote the resource allocated to node $u$.
		We use $\mathbf{r} = \{ r_u \}_{u\in V}$ to denote an \emph{allocation strategy}.
	\end{definition}
	
	\begin{definition}[Reallocation Strategy]
		We use $t(u,v)\geq 0$ to denote the resource $u$ transfers to its neighbor $v$.
		In general $v$ can also send resource to node $u$ (which is denoted by $t(v,u)\geq 0$).
		We use $\mathbf{t} = \{t(u,v),t(v,u)\}_{(u,v)\in E}$ to denote a \emph{reallocation strategy}.
	\end{definition}
	
	The fractions of resource transferred between $u$ and $v$ are upper bounded by the edge weight as follows:
	\[
		t(u,v)\leq w_{uv}\cdot r_u, \qquad t(v,u)\leq w_{uv}\cdot r_v.
	\]
	
	That is, each node $u$ can transfer at most $w_{uv}$ fraction of the resource $r_u$ to its neighbor $v$.
	Additionally, we need to guarantee that the total resources node $u$ sends out is at most the total resource it owns:
	\[
		\textstyle	\sum_{v\in N(u)} t(u,v) \leq r_u.
	\]
	
	Since the resources can be sent and received, the defending power of a node is not fixed.
	Instead, depending on the attack, the defending power at each node can be adaptive by deciding an appropriate reallocation strategy.
	
	\begin{definition}[Defending Power]
		The defending power of node $u$ is defined as the total resource node $u$ owns after the reallocation, which is given as follows:
		\[
			p_u = r_u - \sum_{v\in N(u)} t(u,v) + \sum_{v\in N(u)} t(v,u).
		\]
		We use $\mathbf{p} = \{p_u\}_{u\in V}$ to denote defending powers of nodes.
	\end{definition}
	
	Depending on the reallocation, the defending power $p_u$ of node $u$ can take values in range $[\bar{p}_u,\hat{p}_u]$, where
	\begin{align*}
		\bar{p}_u & = \max\{\textstyle 1-\sum_{v\in N(u)}w_{uv},0 \} \cdot r_u,  \\ 
		\hat{p}_u & = r_u+ \textstyle \sum_{v\in N(u)} w_{uv}\cdot r_v.
	\end{align*}
	
	Note that the allocation strategy $\mathbf{r}$ (which allocates the defending resources) must be decided before the attack happens.
	In contrast, the defender can decide the reallocation strategy depending on which node is attacked.
	Specifically, the defender can define $n$ reallocation strategies $\{ \mathbf{t}^u \}_{u\in V}$, one for each node when it is attacked.
	
	Put differently, there are four sequential steps:
	\begin{itemize}
		\item[(1)] the algorithm decides an allocation strategy $\mathbf{r}$, which allocates a total of $R$ resources;
		\item[(2)] the attacker picks a node $u$ to attack;
		\item[(3)] the algorithm decides a reallocation strategy $\mathbf{t}^u$ to minimize the loss due to the attack. Note that at this point, the allocation strategy is fixed, but the defending power depends on the reallocation strategy.
		\item[(4)] the loss due to the attack is evaluated.
	\end{itemize}
	
	\begin{definition}[Defending Strategy]
		We refer to a solution for the defending problem as a \emph{defending strategy} $(\mathbf{r},\{\mathbf{t}^u\}_{u\in V})$, which consists of an allocation strategy $\mathbf{r}$ and $n$ reallocation strategies $\{ \mathbf{t}^u \}_{u\in V}$.
	\end{definition}

	\subsection{Loss Due to An Attack}
	
	Next, we define the loss due to an attack. Let $\mathbf{p} = \{p_u\}_{u\in V}$ be the defending powers of nodes.
	Suppose $u$ is attacked, the attack spreads to all nodes in $N_k(u)$, where $k$ is a parameter that represents the level of contagiousness of the attack.
	The loss due to the attack is the total damage caused at nodes in $N_k(u)$, where each node $v\in N_k(u)$ suffers from a damage of $\alpha_v$ if $p_v < \lb_v$.
	If $p_v \geq \lb_v$, then no damage is caused at $v$.
	
	\begin{definition}[Defending Result]
		Given defending strategy $(\mathbf{r},\{\mathbf{t}^u\}_{u\in V})$, let $\textsf{Loss}(u)$ be the total damage when $u$ is attacked and the reallocation strategy $\mathbf{t}^u$ is deployed. 
		The defending result is defined as the maximum loss due to an attack, i.e., $\max_{u\in V} \textsf{Loss}(u)$.
	\end{definition}
	
	The objective of the problem is to compute a defending strategy with the minimum defending result.
	We use $\opt$ to denote the optimal (minimum) defending result.
	In the remaining part of the paper, we use \dma (Defending against Contagious Attack) to refer to the problem of computing the defending strategy against contagious attack.
	Note that the decision problem of verifying whether a defending strategy has result at most some value is in \NP.
	Given the defending strategy, the verification can be done by computing $\textsf{Loss}(u)$ for every node $u$ and taking the maximum, both of which take polynomial time.

	\paragraph{Remark.}
	When $k=0$, there is no spreading effect and we only need to protect the node under attack by borrowing defending resources from its neighbors.
	Hence in this case we have $p_u = \hat{p}_u$ if node $u$ is attacked.
	Consequently, the problem degenerates to the single-threshold model of~\cite{aaai/LiTW20}, which can be solved in polynomial time.
	However, in general (when $k\geq 1$), when the attack spreads to multiple nodes, the reallocation must be carefully designed so as to protect multiple nodes, because when a node transfers resource to its neighbors, its own defending power decreases.

	\section{Optimal Response to an Attack}\label{sec:reallocation}
	
	As a warm-up towards further analysis, in this section, we first focus on the subproblem of computing optimal reallocations.
	That is, given a fixed allocation strategy $\mathbf{r} = \{r_u\}_{u\in V}$ and suppose node $u$ is under attack, we compute the reallocation strategy $\mathbf{t}^u$ with which $\textsf{Loss}(u)$ is minimized.
	%
	%
	The following example shows how an appropriate reallocation of resources helps reduce the damage due to an attack.
	
	\begin{example}
		Consider the graph given in Figure~\ref{fig:example-reallocation}(a), and node $a$ is under attack.
		Assuming $k=1$, the attack spreads to $N_1(a) = \{ a,b,d,e \}$.
		Suppose (1) all edges have weight $0.5$; (2) $\lb_a = 4, \lb_b = \lb_d = 2$ and $\lb_e = 3$; and (3) all nodes have defending resource $2$.
		Obviously, without any reallocation, we suffer from a total loss of $\alpha_a+\alpha_e$ since only nodes $b$ and $d$ are sufficiently defended.
		However, if we reallocate the resources as shown in Figure~\ref{fig:example-reallocation}(b), then all nodes in $N_1(a)$ are well defended, and no loss incurs.
		
		\begin{figure}[htb]
			\centering
			\subfigure[]{
				\centering
				\includegraphics[width = 0.22\textwidth]{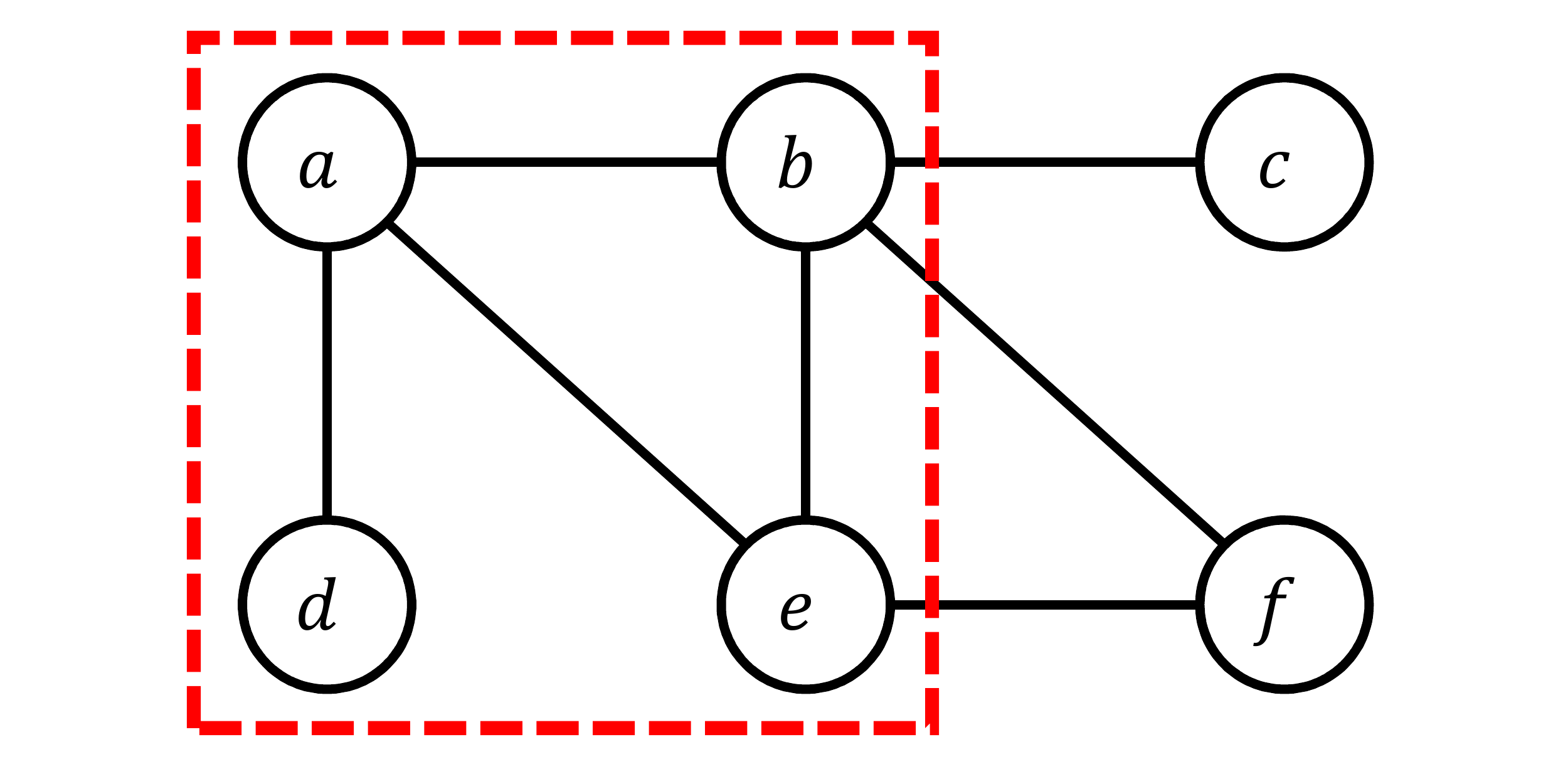}
			}
			\subfigure[]{
				\centering
				\includegraphics[width = 0.22\textwidth]{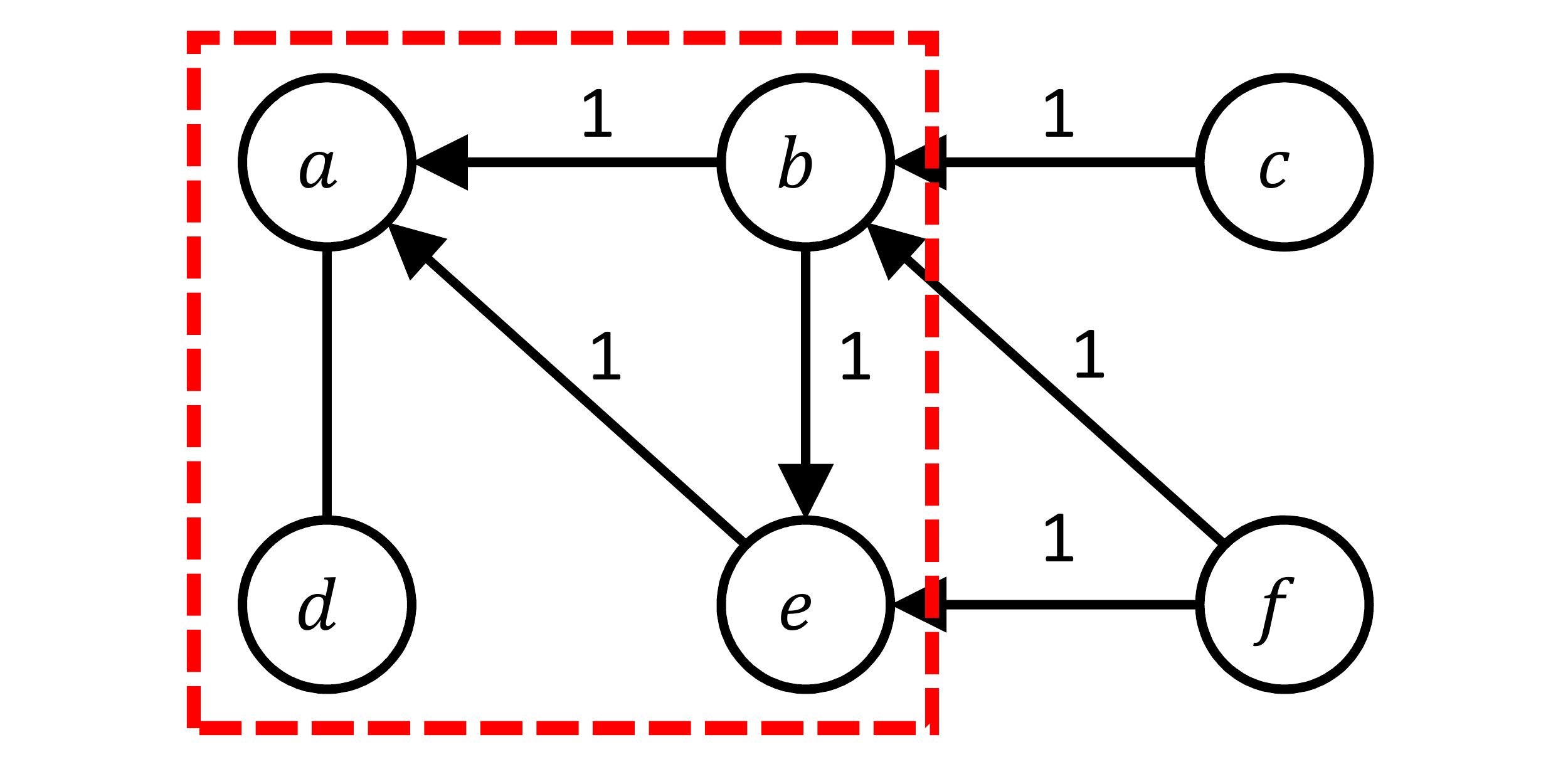}
			}
			\vspace*{-10pt}
			\caption{Example of a reallocation strategy, where a directed edge indicates a transfer of resource. For example, the edge from $b$ to $a$ with value $1$ indicates that node $b$ transfers $t(b,a)=1$ unit of resource to node $a$.}
			\label{fig:example-reallocation}
		\end{figure}
	\end{example}
	
	However, in general, we cannot guarantee that there always exists a reallocation strategy under which all nodes under attack are well defended.
	In this case, we need to compute a reallocation strategy to minimize the total loss.
	For example, we can choose to protect nodes $u$ with larger value $\alpha_u$ while leaving some nodes $v$ with smaller $\alpha_v$ insufficiently defended.
	Unfortunately, we show that the problem of computing the optimal reallocation strategy is \NP-hard.
	For space reasons, we move the proof of the following hardness result to the full version of the paper. 
	
	\begin{theorem}\label{th:hardness}
		Unless \P=\NP, there does not exist any polynomial time algorithm that, given an allocation strategy and a node under attack, computes the optimal reallocation strategy, for any $k\geq 1$.
	\end{theorem}

	
	Next, we formulate the problem of computing the optimal reallocation strategy as a Mixed Integer Linear Program (MILP).
	Recall that we are given an allocation strategy $\mathbf{r}$ and a node $u$ that is attacked.
	\begin{align}
		\text{minimize} \quad \textstyle\sum_{v\in N_k(u)} (1-x_v)\cdot \alpha_v & \nonumber\\
		\text{subject to } \quad \textstyle r_v - \sum_{z\in N(v)}t(v,z) + &\textstyle \sum_{z\in N(v)} t(z,v) \nonumber \\
		\geq \lb_v\cdot x_v, & \quad \forall v \in N_k(u) \label{constraint:reallocation-power} \\
		0\leq t(v,z) \leq w_{vz}\cdot r_v, &\quad \forall z,v\in V \label{constraint:reallocation-edge-bound} \\
		\textstyle \sum_{z\in N(v)}t(v,z) \leq r_v, & \quad \forall v\in V \label{constraint:reallocation-total-bound} \\
		x_v\in \{0,1\}, & \quad \forall v\in N_k(u). \nonumber
	\end{align}
	
	For each node $v\in N_k(u)$ we introduce an integer variable $x_v\in\{0,1\}$ that indicates whether $p_v \geq \lb_v$.
	We introduce fractional variables $t(v,z),t(z,v)$ for each $(v,z)\in E$. 
	The objective of the MILP is the total loss due to the attack, which is the sum of values $\alpha_v$ for $v\in N_k(u)$ that is not well defended ($x_v = 0$).
	Constraints~\eqref{constraint:reallocation-power} guarantee that if we set $x_v = 1$, then $v$ should be well defended, i.e., $p_v \geq \theta_v$.
	Constraints~\eqref{constraint:reallocation-edge-bound} and~\eqref{constraint:reallocation-total-bound} ensure that the transfers of resource between neighboring nodes are feasible.
	
	Note that $\{r_v\}_{v\in V}$ are given and are not variables.
	
	The optimal solution $(\mathbf{x},\mathbf{t})$ for the MILP gives an optimal reallocation $\mathbf{t}$ that minimizes $\textsf{Loss}(u)$, with the fixed allocation $\mathbf{r}$ and node $u$ that is attacked.
	
	\paragraph{Remark.}
	There are redundant variables that can be removed from the MILP.
	Recall that $N_k(u)$ are the nodes the attack spreads to.
	For each $v\in V\setminus N_k(u)$, we have no defending requirements and thus do not need to transfer any resources towards these nodes.
	Consequently, it is unnecessary to introduce variable $t(z,v)$, for any $z\in N(v)$.
	In other words, we only introduce the variable $t(z,v)$	 if $v\in N_k(u)$.
	With this observation, we can reduce the total number of fractional variables from $|E|$ to $\sum_{v\in N_k(u)} |N(v)|$, which is much smaller when $k$ is small and the graph is sparse.
	
	\medskip
	
	Note that the MILP can not be solved exactly in time polynomial in $|N_k(u)|$.
	%
	A natural idea is to relax the integer variables $\mathbf{x}$ to take values in $[0,1]$.
	However, the following instance shows that the integrality gap between the MILP and its LP relaxation is unbounded.
	
	\begin{example}[Integrality Gap] \label{example:integrality-gap}
		Consider the trivial graph with only one node $u$, where $\lb_u=\alpha_u=1$.
		Suppose $R = r_u = 1-\epsilon$, where $\epsilon>0$ is arbitrarily small. 
		Obviously we have $\textsf{Loss}(u) = 1$.
		However, the optimal objective of the LP relaxation is $\epsilon$, by setting $x_u = 1-\epsilon$.
	\end{example}

	\paragraph{Observations.}
	While the integrality gap of MILP and its LP relaxation is unbounded, we still have two useful observations.
	First, the optimal objective of the LP relaxation provides a lower bound on the optimal objective of the MILP, which will be utilized to do a pruning on the MILP in later sections.
	Second, for a fixed $\{0,1\}$-vector $\mathbf{x}\in \{0,1\}^{N_k(u)}$, the MILP becomes a feasibility LP, which can be solved efficiently.
	For example, we use this idea to compute defending strategies with defending result $0$ in Section~\ref{ssec:perfect-allocation}.
	We also extend this idea in Section~\ref{ssec:bicriteria-approx} to compute a polynomial time bi-criteria approximation.
	The idea is to find a vector $\mathbf{x}\in \{0,1\}^{N_k(u)}$ for which the induced LP is feasible, and the objective $\sum_{v\in V}(1-x_v)\cdot \alpha_v$ is as small as possible.

	\section{Computing the Defending Strategy}
	
	In this section, we consider the computation of defending strategies and extend the observations and ideas from the previous section.
	Recall that the defending result is $\max_{u\in V} \textsf{Loss}(u)$, and is uniquely determined by the defending strategy $(\mathbf{r},\{ \mathbf{t}^u \}_{u\in V})$.
	We have shown in Theorem~\ref{th:hardness} that given a fixed allocation strategy and a node under attack, computing the optimal reallocation strategy is \NP-hard.
	However, the hardness result does not necessarily imply a hardness result for computing the allocation strategy.
	In the following, we show that computing the allocation strategy is indeed \NP-hard.
	
	\subsection{Hardness}\label{ssec:hardness-allocation}
	
	We first define a simple special case of the \dma problem called \emph{isolated model}, and then show that even for this special case, the problem is \NP-hard.
	
	\begin{definition}[Isolated Model]
		We refer to the \dma problem where $w_{uv}=0$ for all $(u,v)\in E$ as the \emph{isolated model}.
	\end{definition}
	
	Note that in the isolated model, the defending strategy consists of only an allocation strategy since no reallocation is allowed.
	When $k=0$, the special case can be solved trivially by greedily allocating resources to the nodes with maximum value, because the defending result is defined by the not-well-defended node with maximum value. 
	
	However, in contrast to the case when $k=0$, we show that when $k\geq 1$, the problem becomes \NP-hard.
	
	\begin{theorem}\label{th:hardness-allocation}
		Computing the optimal defending strategy is \NP-hard for $k\geq 1$, even for the isolated model with identical thresholds.
	\end{theorem}
	\begin{proof}
		We prove by a reduction from the (unweighted) vertex cover (VC) problem, which is known to be \NP-hard~\cite{tcs/ChlebikC06}.
		Given an instance $G_{vc} = (V_{vc},E_{vc})$, the VC problem is to select a minimum size subset $S \subseteq V_{vc}$ such that each edge $(u,v)\in E_{vc}$ has at least one endpoint in $S$.
		We construct an instance $G = (V,E)$ of the \dma problem in which $w_{uv} = 0$ for all edges $(u,v)\in E$ and $\lb_u = 1$ for all nodes $u\in V$ as follows.
		
		Let the instance $G$ of the \dma problem be obtained by inserting a node for every edge $(u,v)\in E_{vc}$, splitting the edge.
		Specifically, we first initialize $G = G_{vc}$.
		Then for each $e = (u,v)\in E_{vc}$, we remove $e$, insert a new node $u_e$ and two edges $(u,u_e), (u_e,v)$ into $E$.
		We refer to these nodes (that split edges) the \emph{splitting nodes}, and the other nodes as \emph{original nodes}.
		Note that each splitting node has exactly two neighbors, both of which are original nodes.
		The neighbors of each original node are all splitting nodes.
		Note that we have $|V| = |V_{vc}|+|E_{vc}|$ and $|E| = 2|E_{vc}|$.
		Set $\alpha_u = 0$ for splitting nodes, and $\alpha_u = 1$ for original nodes.
		In other words, only the original nodes are valuable and worth defending.
		Let $\lb_u = 1$ for all $u\in V$ and $k=1$.
		
		Observe that since resource cannot be transferred, the optimal allocation strategy assigns resource either $0$ or $1$ to each original node, and $0$ to each splitting node.
		We call a node $u$ \emph{defended} if $r_u = 1$, \emph{undefended} otherwise.
		Since $k=1$, when the attacker chooses to attack an original node $u$, the total loss is $0$ if $u$ is defended, $1$ otherwise.
		However, if the attacker attacks a splitting node, the total loss is the number of undefended neighbors of the splitting nodes, which can be $2$.
		
		Suppose there exists an allocation strategy using total resource $R$ for which the defending result is at most $1$, then there must exist a vertex cover of size at most $R$ for $G_{vc}$.
		Specifically, the defended original nodes form a vertex cover for $G_{vc}$ (otherwise, there exists a splitting node whose two neighbors are both undefended). 
		Hence if there exists a polynomial time algorithm for the \dma problem, then we can use binary search on $R \in \{1,2,\ldots,|V_{vc}|-1\}$ to identify the minimum $R$ with which the defending result is $1$.
		Consequently, we can compute a minimum vertex cover in polynomial time, which is a contradiction.
	\end{proof}
	
	Interestingly, we show that the reduction also implies a hardness of approximation.
	
	\begin{corollary}
		For any $c<2$, computing a $c$-approximation defending strategy when $k\geq 1$ is \NP-hard, even for the isolated model with identical thresholds.
	\end{corollary}
	\begin{proof}
		Observe that in the above reduction, for any $R < |V_{vc}|$, the defending result is either $1$ or $2$.
		Let $\opt$ be the optimal defending result and $\alg$ be that of the $c$-approximation algorithm, where $c<2$.
		Note that both $\opt$ and $\alg$ take values in $\{1,2\}$.
		Observe that for $\opt = 1$, we must have $\alg = 1$ since otherwise the approximation ratio is $2$.
		Similarly, for $\opt = 2$, we have $\alg = 2$.
		Hence any better-than-$2$ approximation algorithm is equivalent to an exact algorithm, and the corollary follows from Theorem~\ref{th:hardness-allocation}.
	\end{proof}

	\subsection{MILP Formulation}\label{ssec:MILP-main}
	
	Nevertheless, we show that we can formulate the computation of the optimal defending strategy as an MILP as we have done in Section~\ref{sec:reallocation}.
	Similar as before, we introduce a set of variables for the case when $u$ is under attack: we introduce an integer variable $x^u_v\in\{0,1\}$ for each $v\in N_k(u)$, which indicates whether $p_v \geq \lb_v$ when $u$ is under attack; we also introduce a variable $t^u(z,v)$ for each $v\in N_k(u)$ and  $z\in N(v)$, which represents the resource $z$ sends to $v$.
	
	Unlike before, where the allocation strategy is given, here we introduce a variable $r_u$ to denote the resource allocated to node $u\in V$.
	We also changed the objective from minimizing $\textsf{Loss}(u)$ to minimizing $\max_{u\in V} \textsf{Loss}(u)$, by introducing a variable $\textsf{Loss}$ that is at least $\textsf{Loss}(u) = \sum_{v\in N_k(u)} (1-x^u_v) \alpha_v$ for all $u\in V$.
	The computation of the defending strategy is then formulated as follows.
	\begin{align}
		\text{minimize } \quad\qquad \textsf{Loss}\qquad\quad \nonumber& \\
		\text{subject to} \quad \textstyle \sum_{u\in V} r_u \leq R, \quad \nonumber & \\
		\textstyle r_v - \sum_{z\in N(v)\cap N_k(u)} t^u(v,z) +  & \textstyle \sum_{z\in N(v)}t^u(z,v)  \nonumber \\
		\geq \lb_v\cdot x^u_v, \quad & \forall u, v \label{constraint:power} \\
		0\leq t^u(v,z) \leq w_{vz}\cdot r_v,\quad & \forall u, v, z \label{constraint:edge-bound} \\
		\textstyle \sum_{z\in N(v)\cap N_k(u)} t^u(v,z) \leq r_v, \quad & \forall u, v, z \label{constraint:total-bound} \\
		\textstyle\sum_{v\in N_k(u)} (1-x^u_v) \alpha_v \leq \textsf{Loss}, \quad & \forall u \label{constraint:total-loss} \\
		x^u_v\in \{0,1\}, \nonumber \quad & \forall u, v.
	\end{align}
	
	Similar as before, the set of constraints~\eqref{constraint:power} guarantees that the defending power of a node $v$ is at least $\theta_v$ when $x^u_v=1$.
	Constraints~\eqref{constraint:edge-bound} and~\eqref{constraint:total-bound} guarantee feasibility of transfers of resource.
	Constraints~\eqref{constraint:total-loss} ensure $\textsf{Loss} = \max_{u\in V} \textsf{Loss}(u)$ in the optimal solution.
	As before, we only need to introduce variable $t^u(z,v)$ if $v\in N_k(u)$ and $z\in N(v)$.
	We use MILP$(R)$ to denote the above program that uses total resource $R$.
	Note that in the program $r_u$'s and $t^u(z,v)$'s are fractional variables while $x^u_v$'s are integer variables.
	%
	We denote by LP$(R)$ the linear program relaxation when we replace each constraint $x^u_v\in \{0,1\}$ with $x^u_v\in [0,1]$.
	As Example~\ref{example:integrality-gap} shows, the integrality gap of LP$(R)$ and MILP$(R)$ is unbounded.
	Nevertheless, LP$(R)$ provides a lower bound for MILP$(R)$, which can be used for a pruning on MILP.
	
	\paragraph{Prunings.}
	Suppose we have a lower bound $l$ of the optimal defending result $\opt$, i.e., the optimal objective of MILP$(R)$.
	Then for every node $u$ with $\sum_{v\in N_k(u)}\alpha_v \leq l$, we can remove all variables with superscript $u$ and all constraints containing such variables.
	The reason is, when $u$ is attacked, the maximum loss (even if we do not allocate or reallocation any resource) is at most $\sum_{v\in N_k(u)} \alpha_v$.
	Given $\opt \geq l$, not defending nodes in $N_k(u)$ does not increase the objective of MILP$(R)$.
	Note that the optimal solution of LP$(R)$ gives one such lower bound $l$.
	The closer the optimal objectives of LP$(R)$ and MILP$(R)$ are, the better the pruning reduces the size of MILP$(R)$.
	We can further reduce the number of constraints by exploiting the dominance between them.
	For example, if for a node $v$ we have $\sum_{z\in N(v)}w_{vz}\leq 1$ then Constraint~\eqref{constraint:total-bound} of node $v$ will be dominated by Constraints~\eqref{constraint:edge-bound}, and hence can be removed.
	On the other hand, if for a node $v$ we have $w_{vz} = 1$ for all $z\in N(v)$ then Constraints~\eqref{constraint:edge-bound} of node $v$ will be dominated by Constraint~\eqref{constraint:total-bound}, and hence can be removed.

	\subsection{Existence of Perfect Defending Strategy}\label{ssec:perfect-allocation}
	
	While the general problem of computing the optimal allocation strategy is \NP-hard, we show in this section that deciding whether there exists a defending strategy with defending result $0$ (which we refer to as a \emph{perfect defending strategy}) is polynomial time solvable.
	Moreover, if they exist, then we can compute one in polynomial time.
	
	\begin{theorem}\label{th:perfect-allocation}
		For every $k\geq 0$, there exists a polynomial time algorithm that computes a perfect defending strategy for the \dma, if perfect defending strategies exist.
	\end{theorem}	
	\begin{proof}
		Recall that MILP$(R)$ computes the optimal defending strategy.
		If there exist perfect defending strategies, then we have $\textsf{Loss} = 0$ in the optimal solution for MILP$(R)$.
		Since $\textsf{Loss} \geq \sum_{v\in N_k(u)}(1-x^u_v)\alpha_v$, we must have $x^u_v = 1$ for all integer variables in the optimal solution.
		
		Therefore, by fixing $x^u_v = 1$ for all integer variables, MILP$(R)$ must be feasible.
		Observe that after fixing an assignment to the integer variables, MILP$(R)$ becomes a feasibility LP, which can be solved exactly in polynomial time.
		Any feasible solution $(\mathbf{r},\{ \mathbf{t}^u \}_{u\in V})$ for the LP provides a perfect defending strategy, as claimed.
	\end{proof}

	\subsection{Bi-criteria Approximation} \label{ssec:bicriteria-approx}
	
	As Example~\ref{example:integrality-gap} indicates, it is impossible to obtain any bounded approximation of the reallocation by rounding the LP relaxation of MILP$(R)$.
	However, we show that by augmenting the total resource we use, good approximation solutions (in terms of defending results) can be obtained.
	
	\begin{definition}[Bi-criteria Approximation]
	    We call a defending strategy $(\gamma,\beta)$-approximate if it uses $R$ total resource and its defending result is at most $\gamma\cdot \opt$, where $\opt$ is the optimal defending result using $R/\beta$ resource.
	\end{definition}
	
	While it is not possible to obtain bounded (standard) approximations by rounding LP$(R)$, we show that achieving bi-criteria approximations is possible.
	We defer the proof of the following theorem to the full version of the paper.
	
	\begin{theorem}\label{th:bicriteria-approx}
		For any $\epsilon \in (0,1)$, we can compute a $(\frac{1}{1-\epsilon}, \frac{1}{\epsilon})$-approximate defending strategy in polynomial time.
		In particular, with $\epsilon=0.5$ we can compute a $(2,2)$-approximate solution in polynomial time.
	\end{theorem}

	Interestingly, we show that under the Unique Game Conjecture (UGC)~\cite{jcss/KhotR08}, there do not exist strong Pareto improvements over our bi-criteria $(2,2)$ approximation ratio.
	The proof is deferred to the full version of the paper.
	
	\begin{lemma}\label{lemma:ugc-hardness}
		Under UGC, there does not exist polynomial time $(2-\delta,2-\delta)$-approximate algorithm for the \dma problem, for any constant $\delta>0$.
	\end{lemma}

	\paragraph{Implementation.}
	In practice, we can enumerate different $\epsilon\in (0,1)$ to compute different defending strategies, and then pick the one with the best defending result.
	In the following, we show that we might be able to improve the defending result further by deploying a more aggressive rounding on $\mathbf{x}$.
	Specifically, we first solve $\textsf{LP}(\epsilon\cdot R)$ and get the optimal solution.
	Then we pick some $\tau \in [0,\epsilon]$, round each $x$ variable that is less than $\tau$ to $0$, and those at least $\tau$ to $1$.
	With the fixed integer variables, we solve MILP$(R)$, which has become a feasibility LP.
	If the resulting LP is feasible, then we obtain a defending strategy with defending result at most $\frac{1}{1-\tau}\cdot \opt$, where $\opt$ is the optimal defending result of defending strategies using $\epsilon\cdot R$ resources.
	Hence the resulting solution is a $(\frac{1}{1-\tau},\frac{1}{\epsilon})$-approximate defending strategy.
	For different problem instances, the minimum $\tau$ with which the induced LP is feasible can be different.
	However, the LP must be feasible when $\tau = \epsilon$.
	As we will show in our experiments, in all datasets we consider, the defending result after optimizing $\tau$ is much smaller than using $\tau = \epsilon$.

	\section{Experimental Evaluation}\label{sec:experiments}
	
	In this section, we evaluate the effectiveness and efficiency of our algorithms on synthetic and real-world datasets.
	Our datasets contain synthetic graphs, including random graphs and power-law distribution graphs, which are well recognized as the best in modeling random networks and social networks.
	We also consider real-world networks, including aviation networks and social networks, in order to demonstrate the practical performance of our algorithms on defending against contagious attacks in the real world. 
	The datasets are generated as follows. For each dataset, $n$ and $m$ denote the number of nodes and edges, respectively.
	
	
	\begin{table}[ht]
		\centering
		\resizebox{\linewidth}{!}{
			\begingroup
			\renewcommand{\arraystretch}{1.2}
			\begin{tabular}{ c | c c c | c c c }
				\toprule
				& Rand & Pow-S & Pow-L & USAir & FB & Twit \\
				\midrule
				\# Node & 200 & 400 & 700 & 221 & 600 & 1000 \\
				\# Edge & 803 & 1579 & 2087 & 2166 & 4638 & 13476 \\
				\bottomrule
			\end{tabular}
			\endgroup}
		\label{tbl_data}
	\end{table}
	
	\begin{itemize}
		\item \textbf{Random}:
		We generate the dataset with $n=200$ and $p=0.04$ using the algorithm by~\cite{gnp}, where there is an edge between each pair of nodes independently with probability $p$.
		The thresholds $\lb_u$'s and values $\alpha_u$'s are chosen uniformly at random from integers in $[1,10]$.
		The edge weights $w_{uv}$'s are uniformly chosen from $[0.3,1]$.
		
		\item \textbf{Power-law distribution graphs (Pow)}:
		We use the graph generator by NetworkX~\cite{SciPyProceedings_11} to generate the power-law distribution graphs, where we set the parameters\footnote{For the details regarding how the parameters define the graph, please refer to~\url{https://networkx.github.io/documentation/networkx-1.10/reference/generated/networkx.generators.random_graphs.powerlaw_cluster_graph.html}.} to be $(400, 4, 0.5)$ for Pow-S and $(700, 3, 0.5)$ for Pow-L.
		The parameters $\lb_u$'s, $\alpha_u$'s and $w_{uv}$'s are generated randomly as before (for dataset {Rand}).
		
		\item \textbf{USAir}: 
		We select the flight records in the US from years $2008$ to $2010$~\cite{USAir} to generate a directed graph where each node represents a city. There is a directed edge from city $u$ to city $v$ if the number of flights per week from $u$ to $v$ is at least $25$.
		We set the edge weight as the ratio between the flights-per-week of the edge and the maximum flights-per-week value of all edges.
		We set $\lb_u$ and $\alpha_u$ as the population (in millions) of city $u$.
		
		\item \textbf{Social networks}:
		We use the network of Facebook (undirected) and Twitter (directed) to generate our datasets~\cite{fb&twit}.
		The dataset FB (resp. Twit) is extracted from the source network by picking a random node in the network and expand using breath-first-search until the size of the dataset reaches $n = 600$ (resp. $n=1000$).
		We set $\lb_u = \alpha_u = w_{uv} = 1$ for all nodes and edges.
	\end{itemize}
	
	\subsubsection{Experiment Environment.} We perform our experiments on an AWS Ubuntu 18.04 machine with 32 threads and 128GB RAM without GPU. We use Gurobi optimizer~\cite{gurobi} as our solver for the LPs and MILPs.
	
	\medskip
	
	We evaluate the effectiveness of our exact and approximation algorithms by comparing the results of defending strategies under different settings and datasets.
	Throughout all the experiments, we fix the contagiousness parameter $k=2$.
	
	\subsubsection{Effectiveness of Reallocation.} One of the main innovations of our work is that we consider the reallocation of defending resources between the nodes.
	The reallocation allows the algorithm to react adaptively against the attack.
	In particular, we compare the results of defending strategies with and without reallocation as follows.
	
	We first use the algorithm in Section~\ref{ssec:perfect-allocation} to compute for each dataset the minimum total resource required in a perfect defending strategy (a strategy with defending result $0$).
	As our experiment (in Table~\ref{tbl_reall_1}) shows, reallocation (see the row with $w\neq 0$) always helps in reducing the requirement on the defending resource, for all datasets.
	For example, for the first dataset Rand, the resource required in a perfect defending strategy is 40\% less than the case when reallocation is not allowed (see the row with $w=0$).
	
	\begin{table}[ht]
		\centering
		\resizebox{\linewidth}{!}{
			\begingroup
			\renewcommand{\arraystretch}{1.2}
			\begin{tabular}{ c | c c c c c c}
				\toprule
				& Rand & Pow-S & Pow-L & USAir & FB & Twit \\
				\midrule
				$w = 0$ & 1037 & 1892 & 3406 & 341 & 600 & 1000 \\
				$w \neq 0$ & 587 & 1687 & 2320 & 340 & 524 & 623 \\
				\bottomrule
			\end{tabular}
			\endgroup}
		\caption{Resource required to achieve defending result $0$.}
		\label{tbl_reall_1}
	\end{table}

	\vspace*{-5pt}
	
	\subsubsection{Comparing Different Algorithms.} We also evaluate the effectiveness of our bi-criteria approximations from Section~\ref{ssec:bicriteria-approx}, and compare it with the exact solution and the Greedy algorithms.
	The results are shown in Table~\ref{tbl_eptau}, where BA$(\epsilon)$ stands for the approximation algorithm by rounding the optimal solution for LP$(\epsilon\cdot R)$ and optimizing $\epsilon\in(0,1)$; BA$(\epsilon,\tau)$ stands for the approximation algorithm that further optimizes $\tau\in ( 0,\epsilon ]$ in the more aggressive rounding.
	We use Greedy to refer to the algorithm that greedily allocates resources to nodes with the maximum value (break tie arbitrarily) and does not use reallocation; Greedy-R is based on Greedy but uses greedy reallocation.
	Specifically, when node $u$ is attacked, for each node $v\in N_k(u)$ in decreasing order of their values, the algorithm transfers resource to $v$ until its defending power is at least its threshold, or when no more resource can be transferred from its neighbors.	
	In the experiments, we fix $R = 0.5\cdot \sum_{u\in V} \lb_u$ for all datasets.
	
	\begin{table}[ht]
		\centering
		\resizebox{\linewidth}{!}{
			\begingroup
			\renewcommand{\arraystretch}{1.2}
			\begin{tabular}{ c | c c c c c c}
				\toprule
				& Rand & Pow-S & Pow-L & USAir & FB & Twit \\
				\midrule
				Greedy & 278 & 859 & 1168 & 178.7 & 188 & 302 \\
				Greedy-R & 225 & 819 & 1025 & 178.7 & 186 & 281 \\
				BA$(\epsilon)$ & 289 & 1134 & 1230 & 291.8 & 109 & 148 \\
				BA$(\epsilon,\tau)$ & 107 & 785 & 701 & 204.3 & 58 & 55 \\
				\midrule
				Exact & 69 & 740 & 616 & 170.3 & 51 & 53 \\
				\bottomrule
			\end{tabular}
			\endgroup}
		\caption{Defending results by the bi-criteria approximations.}
		\label{tbl_eptau}
	\end{table}
	
	\noindent From Table~\ref{tbl_eptau}, we observe that by deploying a more aggressive rounding on the fractional solutions, the approximation solutions by BA$(\epsilon,\tau)$ outperform BA$(\epsilon)$ dramatically, and are very close to the optimal solutions in all datasets.
	Moreover, in general, BA$(\epsilon,\tau)$ achieves much smaller defending results when compared with both Greedy approaches in most datasets.
	The only exception is the dataset USAir, in which the values of nodes differ greatly. Thus, Greedy allocation of resources achieves the almost optimal result.
	BA$(\epsilon)$ does not perform well because the solution is obtained by a very loose rounding on the solution of LP$(\epsilon\cdot R)$.
	Observe that with the help of reallocation, Greedy-R obtains advantages over Greedy, which again demonstrates the critical role of reallocation.
	Note that the defending results by Pow-S are larger than that of Pow-L because it has a larger density of edges, which crucially affects the number of attacked nodes.

	\subsubsection{Efficiency Evaluation.}
	Finally, we evaluate the efficiency of our algorithms and summarize the running times in Table~\ref{tbl_rtime}, where No-Prune refers to the algorithm by solving the MILP without using the pruning we mentioned in Section~\ref{ssec:MILP-main}; Pruning refers to the one with pruning; BA$(\epsilon,\tau)$ refers to our bi-criteria approximation algorithm.
	
	\begin{table}[ht]
		\centering
		\resizebox{\linewidth}{!}{
			\begingroup
			\renewcommand{\arraystretch}{1.2}
			\begin{tabular}{ c |c c c c c c}
				\toprule
				& Rand & Pow-S & Pow-L & USAir & FB & Twit \\
				\midrule
				No-Prune & 4637 & 8180 & 36293 & 10686 & 32608 & 8441 \\
				Pruning & 4259 & 3351 & 8227 & 2920 & 18929 & 2757 \\
				BA$(\epsilon,\tau)$ & 246 & 867 & 1508 & 416 & 2540 & 1294 \\
				\bottomrule
			\end{tabular}
			\endgroup}
		\caption{Running times of different algorithms (in seconds).}
		\label{tbl_rtime}
	\end{table}
	
	As we can see from the Table~\ref{tbl_rtime}, the running times for solving MILPs are obviously improved after pruning, which shows the effectiveness of removing redundant variables.
	In particular, for the dataset Pow-L, which admits the long tail phenomenon, there is a 77\% improvement on the running time after pruning.
	Our approximation algorithm improves the running time even further, e.g., is several times faster than that of Pruning in all datasets, which demonstrates its great efficiency in practical use.
	Recall that from Table~\ref{tbl_eptau}, the defending results of our approximation algorithms are very close to the exact solutions by MILP.
	
	\clearpage
	\bibliography{ref}

	\clearpage
	
	\appendix
	
	\section{Missing Proofs}
	
	\begin{proofof}{Theorem~\ref{th:hardness}}
		We prove this by a reduction from the maximum independent set (MIS) problem, which is \NP-hard~\cite{mst/BermanF99}.
		In the MIS problem we are given a graph ${G_{mis}=(V_{mis},E_{mis})}$ and the problem is to find a maximum size set $S \subseteq V_{mis}$ such that no two vertices in $S$ are adjacent.
		We construct an instance of the \dma problem as follows.
		We first initialize the graph structure ${G=(V,E)}$ to be the same as $G_{mis}$, and set the parameters as follows.
		Let $n=|V_{mis}|$ and $d(u) = N(u)$ be the degree of node $u$ in $G$.
		\begin{itemize}
			\item Let $r_v = n$, $\theta_v=n+d(v)$ and $\alpha_v=1$ for all $v \in V$.
			\item Let $w_{uv} = \frac{1}{n}$ for all $(u,v) \in E$.
		\end{itemize}
		
		Finally, we insert a new vertex $s$ to $V$ with $r_{s}=\alpha_{s}=0$, $\theta_s=1$, and let $s$ be connected with all other nodes with edges with weight $0$.
		Note that the final instance of \dma has $n+1$ nodes and $n+|E_{mis}|$ edges.
		
		Suppose $s$ is under attack.
		Since $N_1(s) = V$, the attack spreads to the whole graph $G$, for $k\geq 1$.
		
		Observe that for each node $u\in V_{mis}$, its resource $r_u = n$ is lower than its threshold $\theta_u = n+d(u)$.
		Moreover, its maximum possible defending power
		\begin{equation*}
			\hat{p}_u = n + \sum_{v\in N(u)} w_{uv}\cdot r_v = n + d(u) = \theta_u
		\end{equation*}
		can be obtained only by (1) not transferring any resource to its neighbors; (2) receive $w_{uv}\cdot r_v = 1$ unit of resource from each of its neighbors.
		
		Hence under any reallocation strategy, if a node $u$ is well defended ($p_u \geq \theta_u$), then none of its neighbors $v\in N(u)$ is well defended ($p_v < \theta_v$).
		Let $S^*\subseteq V_{mis}$ be the set of nodes that are well defended under the optimal reallocation strategy, we have
		\begin{itemize}
			\item $S^*$ is an independent set (by the above argument);
			\item $\textsf{Loss}(s) = n - |S^*|$, since each $u\in V_{mis}$ has $\alpha_u = 1$.
		\end{itemize}
		
		Since $\textsf{Loss}(s)$ is minimized in the optimal reallocation strategy, we know that $|S^*|$ is maximized.
		In other words, $S^*$ is a maximum size independent set of $G_{mis}$.
		
		Since the reduction is polynomial time, if there exists a polynomial time algorithm to compute the optimal reallocation strategy, then we can solve the MIS problem in polynomial time, which is a contradiction.
		Consequently, computing the optimal reallocation strategy is \NP-hard.
	\end{proofof}
	
	\begin{proofof}{Lemma~\ref{lemma:ugc-hardness}}
		We use the same reduction from vertex cover problem as in the proof of Theorem~\ref{th:hardness-allocation}.
		It is shown in~\cite{jcss/KhotR08} that under the Unique Game Conjecture (UGC), there does not exist $(2-\delta)$-approximation for the vertex cover problem, for any constant $\delta>0$.
		Suppose there exists a polynomial time $(2-\delta,2-\delta)$-approximate algorithm for the \dma problem, we show that it can be transformed into a $(2-\delta)$-approximation algorithm for the vertex cover problem, which contradicts the UGC.
		
		Given any instance $G_{vc} = (V_{vc},E_{vc})$ of a VC problem instance, we construct an instance $G = (V,E)$ of \dma as in the proof of Theorem~\ref{th:hardness-allocation}.
		Then for each $R = 1,2,\ldots,|V_{vc}|-1$, we run the $(2-\delta,2-\delta)$-approximate algorithm to compute a defending strategy, and let $R^*$ be the smallest such that when $R = R^*$, the defending result is $1$.
		
		Suppose $S^*\subseteq V_{vc}$ is the minimum size vertex cover.
		By the construction of the \dma problem, when $R = |S^*|$, the optimal defending result is $\opt = 1$.
		Hence with $R = (2-\delta)|S^*|$ total defending resource, the $(2-\delta,2-\delta)$-approximate algorithm computes a defending strategy with defending result at most $(2-\delta)\cdot \opt = 2-\delta$.
		Since defending results are integers, the defending result by the approximation algorithm is $1$. 
		Observe that since all nodes have threshold $1$ and resource cannot be transferred, a defending strategy with $R = (2-\delta)|S^*|$ is equivalent to one with $R = \lfloor(2-\delta)|S^*|\rfloor$.
		In other words, when $R = \lfloor(2-\delta)|S^*|\rfloor$, the defending result of the approximation algorithm is $1$, which implies $R^* \leq \lfloor(2-\delta)|S^*|\rfloor$.
		Moreover, since the defending result is $1$, the set of defended nodes is a vertex cover.
		Hence we have found a vertex cover of size
		\begin{equation*}
			R^* \leq \lfloor(2-\delta)|S^*|\rfloor \leq (2-\delta)|S^*|,
		\end{equation*}
		which gives a $(2-\delta)$-approximation for the VC problem, and contradicts the UGC.
	\end{proofof}
    
	\begin{proofof}{Theorem~\ref{th:bicriteria-approx}}
		Recall that by the definition of bi-criteria approximations, we need to compute a strategy with the defending result at most $\frac{1}{1-\epsilon}\cdot \opt$, where $\opt$ is the optimal defending result of defending strategies that use $\epsilon\cdot R$ total defending resource, i.e., $\opt$ is the optimal objective of MILP$(\epsilon\cdot R)$.
		
		We first run $\textsf{LP}(\epsilon\cdot R)$ and obtain the optimal (fractional) solution.
		Note that the optimal objective of $\textsf{LP}(\epsilon\cdot R)$ is at most $\opt$, but in the solution each $x^u_v$ can take arbitrary values in $[0,1]$.
		In the following, we round the optimal solution $(\mathbf{x},\mathbf{r},\mathbf{t})$ of LP$(\epsilon\cdot R)$ and construct a feasible solution $(\hat{\mathbf{x}},\hat{\mathbf{r}},\hat{\mathbf{t}})$ for MILP$(R)$.
		We show that the objective of the solution is at most $\frac{1}{1-\epsilon}\cdot \opt$.
		Since each feasible solution for MILP$(R)$ corresponds to a defending strategy, the theorem follows.
		For each variable $x^u_v\in[0,1]$, let $\hat{x}^u_v = 1$ if $x^u_v \geq \epsilon$, and let $\hat{x}^u_v = 0$ otherwise.
		Let $\hat{r}_u = \frac{1}{\epsilon}\cdot r_u$ and $\hat{t}^u(z,v) = \frac{1}{\epsilon}\cdot t^u(z,v)$, for the corresponding variables.
		We show that the solution $(\hat{\mathbf{x}},\hat{\mathbf{r}},\hat{\mathbf{t}})$ we have constructed is feasible for MILP$(R)$.
		Since originally $\sum_{u\in V}r_u\leq \epsilon\cdot R$, we have $\sum_{u\in V} \hat{r}_u \leq R$, i.e., the first constraint of MILP$(R)$ is satisfied.
		Constraints~\eqref{constraint:power} in which $\hat{x}^u_v = 0$ are trivially satisfied.
		For those with $\hat{x}^u_v = 1$, since we increase $x^u_v$ by a factor of at most $\frac{1}{\epsilon}$ and increase all $r$ and $t$ variables by a factor of $\frac{1}{\epsilon}$, Constraints~\eqref{constraint:power} of MILP$(R)$ are satisfied.
		Since all $r$ and $t$ variables are scaled by the same factor, Constraints~\eqref{constraint:edge-bound}~\eqref{constraint:total-bound} of MILP$(R)$ are all satisfied.
		
		Finally, since we decrease each $x^u_v$ (to $0$) only if $x^u_v < \epsilon$, 	
		by rounding $\mathbf{x}$ into $\hat{\mathbf{x}}$, for each $u\in V$, we have
		\[
		\textstyle	\sum_{v\in N_k(u)} (1-\hat{x}^u_v)\alpha_v \leq \frac{1}{1-\epsilon} 
			\sum_{v\in N_k(u)} (1-x^u_v)\alpha_v.
		\]
		
		Hence the objective of solution $(\hat{\mathbf{x}},\hat{\mathbf{r}},\hat{\mathbf{t}})$ for MILP$(R)$, $\textsf{Loss} = \max_{u\in V}( \sum_{v\in N_k(u)} (1-\hat{x}^u_v)\alpha_v )$, is at most $\frac{1}{1-\epsilon}\cdot \opt$, as claimed.
	\end{proofof}
	
\end{document}